\documentclass[11pt]{article}
 \usepackage[top=0.75in, bottom=0.75in, left=1in, right=1in]{geometry}
     \usepackage{amsmath}
     \usepackage{amssymb}
     \usepackage{amsthm}
     \usepackage{amsmath}
     \usepackage{setspace}
     \usepackage{fancyhdr}
     \usepackage{xcolor}
     \usepackage{amssymb}
     \usepackage{amsthm}
     \usepackage{cite}
     \usepackage{tabularx}
     \usepackage{graphicx}
     \usepackage{epstopdf}
     \usepackage{epsfig}
     \usepackage{float}
     \usepackage{ listings}
   \usepackage{amsmath}
\usepackage{amsthm}
\usepackage{setspace}
\usepackage{xcolor}
\usepackage{fancyhdr}
\usepackage{amssymb}
\usepackage{listings}
 \usepackage{cite}
 \usepackage{url}
  \usepackage{tabularx}
  \usepackage{graphicx}
  \usepackage{epstopdf}
  \usepackage{epsfig}
  \usepackage{float}
  \usepackage{ listings}
  \usepackage{algorithm}
\usepackage{algorithmic}

 \theoremstyle{plain}
 \newtheorem{thm}{Theorem}[section]
 
 \newtheorem{lem}[thm]{Lemma}
 
 \theoremstyle{definition}
 \newtheorem{definition}[thm]{Definition}
 
 \theoremstyle{remark}
 \newtheorem{remark}[thm]{Remark}

 \let\b=\beta

 \let\e=\varepsilon

 \let\lam=\lambda

 \let\f=\frac

 \let\om=\omega

 \let\wt=\widetilde
 \let\td = \tilde

 \let\teq \triangleq

 \def\la{\langle}
 \def\ra{\rangle}

 \def\lt{\left}
 \def\rt{\right}
 \def\B{\Big}

 \newcommand{\beq}{\begin{equation}}
 \newcommand{\eeq}{\end{equation}}
  \newcommand{\bal}{\begin{aligned} }
  \newcommand{\eal}{\end{aligned}}
 \newcommand{\ben}{\begin{eqnarray}}
 \newcommand{\een}{\end{eqnarray}}
 \newcommand{\beno}{\begin{eqnarray*}}
 \newcommand{\eeno}{\end{eqnarray*}}

   \date{\today}
    \author{Jiajie Chen \footnote{ School of Mathematical Sciences, Peking University; Email: cjiajie@pku.edu.cn}, Anthony Hou \footnote{Department of Statistics, Harvard University; Email: ahou@college.harvard.edu}, Thomas Y. Hou \footnote{Applied and Computational Mathematics, Caltech; Email: hou@cms.caltech.edu}}
    {\color{blue}\title{A Prototype Knockoff Filter for Group Selection with FDR Control}}
   
\begin{document}
    \maketitle{}
 \vspace{-0.2in}
\begin{abstract}
In many applications, we need to study a linear regression model that consists of a response variable and a large number of potential explanatory variables and determine which variables are truly associated with the response. In \cite{Candes}, the authors introduced a new variable selection procedure called the knockoff filter to control the false discovery rate (FDR) and proved that this method achieves exact FDR control. In this paper, we propose a prototype knockoff filter for group selection by extending the Reid-Tibshirani prototype method \cite{Tib}. Our prototype knockoff filter improves the computational efficiency and statistical power of the Reid-Tibshirani prototype method when it is applied for group selection. 
In some cases when the group features are spanned by one or a few hidden factors, we demonstrate that the PCA prototype knockoff filter outperforms the Dai-Barber group knockoff filter \cite{Barber}.
 We present several numerical experiments to compare our prototype knockoff filter with the Reid-Tibshirani prototype method and the group knockoff filter. We have also conducted some analysis of the knockoff filter. Our analysis reveals that some knockoff path method statistics, including the Lasso path statistic, may lead to loss of power for certain design matrices and a specially designed response even if their signal strengths are still relatively strong.
    \end{abstract}
    \section{Introduction}

 In many scientific endeavors, we need to determine from a response variable together with a large number of potential explanatory variables which variables are truly associated with the response. In order for this study to be meaningful, we need to make sure that the discoveries are indeed true and replicable.  Thus it is highly desirable to obtain exact control of the false discovery rate (FDR) within a certain prescribed level. In \cite{Candes}, Barber and Cand\`es introduced a new variable selection procedure called the knockoff filter to  control the FDR for a linear model. This method achieves exact FDR control in finite sample settings and does not require any knowledge of the noise level. A key observation is that by constructing knockoff variables that mimic the correlation structure found within the existing variables one can obtain accurate FDR control. The method is very general and flexible. It can be applied to a number of statistics and has more statistical power (the proportion of true signals being discovered) than existing selection rules in some cases.

 \subsection{A brief review of the knockoff filter}\label{sec:rev}
Throughout this paper, we consider the following linear regression model $y=X\beta+ \epsilon$, where the feature matrix $X$ is an $n\times p$ ($ n \geq 2 p$) matrix with full rank, its columns are normalized to be unit vectors in the $l^2$ norm, and $\epsilon$ is a Gaussian noise $\e \sim N(0,\sigma^2 I_n)$. We first provide a brief overview of the knockoff filter introduced in \cite{Candes}. 
The knockoff filter begins with the construction of a knockoff matrix $\tilde{X}$ that obeys
   \begin{equation}\label{revko1}
      \tilde{X}^T\tilde{X} = X^TX, \quad \tilde{X}^T X = X^T X - diag(s) ,
   \end{equation}
where $s_i \in [0,1]$. The positive definiteness of the Gram matrix $[X \tilde{X}]^T[X \tilde{X}]$ requires
   \begin{equation}\label{revko2}
   diag(s) \preceq 2X^TX.
   \end{equation}

The first condition in (\ref{revko1}) ensures that $\tilde{X}$ has the same covariance structure as the original feature matrix $X$. The second condition in (\ref{revko1}) guarantees that the correlations between distinct original and knockoff variables are the same as those between the originals. To ensure that the method has good statistical power to detect signals, we should choose $s_j$ as large as possible to maximize the difference between $X_j$ and its knockoff $\tilde{X}_j$. Once $diag(s)$ is obtained, $\td{X}$ can be constructed in terms of $X, diag(s)$ and an orthonormal matrix $U\in R^{n \times p}$ with $U^TX=0$. The existence of $U$ requires $n \geq 2p$. The next step is to calculate a statistic, $W_j$, for each pair $X_j,\tilde{X}_j$ using the Gram matrix $[X \ \tilde{X}]^T [X \tilde{X}]$ and the marginal correlation $[X \ \tilde{X}]^T y$. In addition, $W_j$ satisfies a flip-coin property, which implies that swapping arbitrary pair $X_j, \td{X}_j$ only changes the sign of $W_j$ but keeps the sign of other $W_i$ ($i \neq j$)  unchanged. The construction of the knockoff features and the symmetry of the test statistic are important to achieve a crucial property that the signs of the $W_j$'s are i.i.d. random for the null hypotheses. This property plays a crucial role in obtaining exact FDR control by using a supermartingale argument.

The final step is to run the knockoff (knockoff+) selection procedure  at the target FDR level $q$. A large positive $W_j$ gives evidence that variable $j$ is a nonnull. 
Then select the model $\hat{S} \triangleq \{j: W_j \geq T  \}$ using a data-dependent threshold $T$ defined below:
  \begin{equation}\label{cex5}
   \begin{aligned}
 &T_i \triangleq   \min \left\{ t>0 :        \frac{ i+  \#\{   j: W_j   \leq -t              \}  }{      \#\{   j: W_j   \geq t\}   \vee 1   }  \leq q  \right\}  ,\quad i = 0, 1 \; . \\
   \end{aligned}
\end{equation}
$T_0$ and  $T_1$ are used in the knockoff and knockoff+ selection procedure, respectively.

There are several ways to construct a statistic $W_j$. Among them, the Lasso path statistic is discussed in detail in \cite{Candes}. It first fits a Lasso regression of $y$ on $[X \tilde{X}]$ for a list of regularizing parameters $\lambda$ in a descending order and then calculates the first $\lambda$ at which a variable enters the model, i.e. $Z_j \triangleq \sup\{ \lambda: \hat{\beta}_j(\lambda) \ne 0 \}$ for feature $X_j$ and   $\tilde{Z}_j = \sup\{ \lambda  :\tilde{\beta}_j(\lambda) \ne 0 \}$ for its knockoff $\tilde{X}_j$. The Lasso path statistic is defined as $W_j = \max(Z_j , \tilde{Z}_j) \cdot sign(Z_j - \tilde{Z}_j)$. If $X_j$ is a nonnull, it has a non-trivial effect on $y$ and should enter the model earlier than its knockoff $\td{X}_j$, resulting in a large positive $W_j$. In this case, the corresponding feature is likely to be selected by the knockoff filter \eqref{cex5}.  If $X_j $ is a null, it is likely that $X_j$ enters the model later, resulting in a small positive or negative $W_j$. The corresponding feature is likely to be rejected according to \eqref{cex5}.

The main result in \cite{Candes} is that the knockoff procedure and knockoff+ procedure has exact control of mFDR and FDR respectively,
\[
mFDR \triangleq E \left[   \frac{   \# \{ j \in \hat{S}: \beta_j =0 \} }{ \# \{ j \in \hat{S}\} + q^{-1} }  \right]  \leq q\;, \quad FDR \triangleq E \left[   \frac{   \# \{ j \in \hat{S}: \beta_j =0 \} }{ \# \{ j \in \hat{S}\} \vee 1 }  \right]  \leq q\; .
\]

In a subsequent paper \cite{Candes2}, Barber and Cand\`es developed a framework for high-dimensional linear model with $p \ge n$.  In this framework, the observations are split into two groups, where the first group is used to screen for a set of potentially relevant variables, whereas the second is used for inference over this reduced set of variables. The authors also developed strategies for leveraging information from the first part of the data at the inference step for greater power. They proved that this procedure controls the directional false discovery rate (FDR) in the reduced model controlling for all screened variables. 

The knockoff filter has been further generalized to the model-free framework in \cite{Candes3}. Whereas the knockoffs procedure is constrained to linear models, 
the model-X knockoffs provide valid inference from finite samples in settings in which the conditional distribution of the response is arbitrary and completely unknown. Furthermore, this holds independent of the number of covariates. They achieved correct inference by constructing knockoff variables probabilistically instead of geometrically. The rigorous FDR control of the model-free knockoffs procedure is also established.

The research development for the knockoff filter has inspired a number of follow-up works, see e.g. \cite{Tib,Barber,Su1,Su2}.

\subsection{A prototype knockoff filter for group selection}
Group selection is an effective way to perform statistical inference when features within each group are highly correlated but the correlation among different groups is relatively weak. Inspired by the prototype method developed by Reid and Tibshirani in \cite{Tib}, we propose a prototype knockoff filter for group selection that has exact group FDR control (defined in Theorem \ref{ptfdr}) for strongly correlated features. Assume that $X$ can be clustered into $k$ groups $X = (X_{C_1},X_{C_2},...,X_{C_k})$ in such a way that the within-group correlation is relatively strong. As in \cite{Tib}, we split the data $(X, y)$ by rows into two disjoint parts $(X^{(1)}, y^{(1)}) $ and $(X^{(2)}, y^{(2)})$, and
extract prototype $X_{P_i}$ for each group $|C_i|$ using the first part of the data  $(X^{(1)}, y^{(1)}) $. We then construct the knockoff matrix $\td{X}_{P}^{(2)}$ only for the prototype features in the second part of the design matrix and run the knockoff selection. 
Finally, select group $i$ if $P_i$ is selected by the knockoff filter.
We have also developed a \textit{PCA prototype filter} and proved that both of these prototype knockoff filters have exact group FDR control. We compare these two prototype filters with a variant of the Reid-Tibshirani prototype method for group selection and the Dai-Barber group knockoff filter \cite{Barber} and provide numerical experiments to demonstrate the effectiveness of our methods.

 \subsection{Alternating sign effect}
 In this paper, we have also performed some analysis of the knockoff filter with certain path method statistics such as the Lasso path and the forward selection statistics. According to \eqref{cex5}, the knockoff filter threshold $T$ is determined by the ratio of the number of large positive and negative $W_j$.  Large, negative $W_j$'s may result in a large $T$ and fewer selected features. For certain design matrix $X$ and a specially designed response $y$ with strong signal strengths, our analysis shows that for the knockoff filter with certain path method statistics, e.g. the Lasso path statistic,
some knockoff variable $\td{X}_j$ can enter the model earlier than its original feature $X_j$. This could lead to large negative $W_j$ and reduce the power. We discuss some possible mechanism under which the path method statistic may suffer from this potential challenge for certain design matrices and the response. A possible scenario is when some features are positively (negatively) correlated but their contribution to the response $y$ has the opposite (same) sign, e.g. $X_j^T X_k > 0,  \b_j > 0, \b_k < 0$. We call this mechanism the \textit{alternating sign effect}. In general, the chance that such potential challenge arises is quite rare.  But it gives a healthy warning that this potential challenge could occur for certain statistics and we must use them with care.

\subsection{Extension of the sufficiency property}
The sufficiency property of a knockoff statistic $W$ in \cite{Candes} states that $W$ depends only on the Gram matrix $[X \tilde{X}]^T[X \tilde{X}]$ and the feature-response product $[ X \tilde{X}]^T y$.  In this definition, only part of the information of the response variable $y$, i.e. $[X \tilde{X}]^Ty$, is utilized. We generalize the sufficiency property such that $W$ can depend on the remaining information of $y$.  This generalization maintains the FDR control of the knockoff procedure.  As an application, we show that the classical noise estimate obtained by using the least squares can be incorporated in the knockoff filter without violating FDR control. We remark that in a recent work \cite{Bien18} the sufficiency property is also relaxed by allowing $W$ to be a function of $|| y ||_2$. The definition we propose is more general than the one used in \cite{Bien18}.

The main motivation for us to study the prototype knockoff filter for group selection is to alleviate the difficulty in feature selection of highly correlated features. In a related work \cite{CHH18}, we have developed a pseudo knockoff filter in which we relax one of the knockoff constrains. Our numerical study indicates that the pseudo knockoff filter could give high statistical power for certain statistic when the features have relatively strong correlation. Although we cannot establish rigorous FDR control for the pseudo knockoff filter as the original knockoff filter, we provide some partial analysis of the pseudo knockoff filter with the half Lasso statistic and establish a uniform FDP bound and an expectation inequality.

The rest of the paper is organized as follows. In Section \ref{gppca}, we introduce our prototype knockoff filter for highly correlated features.
We compare it to other group selection methods and provide numerical experiments to demonstrate the performance of various methods.
In Section \ref{sec_ana}, we discuss the potential challenge of the knockoff filter with certain path method statistics due to the alternating sign effect and generalize the sufficiency property of a knockoff statistic.

 \section{Prototype knockoff filters}\label{gppca}
 In this section, we propose two prototype group selection methods with group FDR control to overcome the difficulty associated with strong within-group correlation. It is well known that the grouping strategy provides an effective way to handle strongly correlated features. Our work is inspired by Reid-Tibshirani's prototype method \cite{Tib}, Dai-Barber's group knockoff filter \cite{Barber} and Barber-Cand\`es' high-dimensional  knockoff filter \cite{Candes2}. We provide a brief summary of the first two methods below before introducing our prototype filters. 

\subsection{Reid-Tibshirani's prototype method}
In \cite{Tib}, Reid and Tibshirani introduced a prototype method for prototype selection. It can be applied directly to group selection and  consists of the following steps. First, cluster columns of $X$ into $K$ groups, $\{C_1,...,C_K\}$. Then split the data by rows into two (roughly) equal parts  $
y = \left(
\begin{array}{c}
 y^{(1)} \\
 y^{(2)} \\
 \end{array}
 \right)
 $ and
$X= \left(  \begin{array}{c}  X^{(1)} \\  X^{(2)} \\  \end{array}  \right) $. Choose a prototype for each group via the maximal marginal correlation, using only the first part of the data $y^{(1)} , X^{(1)}$. This generates the prototype set $\hat{P}$. Next, form a knockoff matrix $\tilde{X}^{(2)}$ from $X^{(2)}$ and  perform the knockoff filter using $y^{(2)}, [X_{\hat{P}}^{(2)} \ \tilde{X}_{\hat{P}}^{(2)}  ]$. Finally, group $C_i$ is selected if and only if $X_{ \hat{P}_i}^{(2)}  $ is chosen in the filter process. The group FDR control is a direct result of Lemma 6.1 \cite{Tib}.
We remark that strong within-group correlation results in small difference between the prototype
$X_{P_i}$ and its knockoff pair $\td{X}_{P_i}$. Suppose that $X_{P_i}$ and  $X_j, \  j \neq P_i$ are in the same group  and strongly correlated, i.e. $|| X_{P_i} - X_j ||_2$ is small ($X_j, X_{P_i}$ are normalized). The knockoff constraint $(X_{P_i} - \td{X}_{P_i})^T X_j = 0$ implies
\beq\label{eq:R-T}
|| \td{X}_{P_i} - X_j ||_2 = || X_{P_i} - X_j ||_2 , \quad || X_{P_i} - \td{X}_{P_i} ||_2 \leq || \td{X}_{P_i} - X_j ||_2 + || X_{P_i} - X_j ||_2 = 2 || X_{P_i} - X_j ||_2.
\eeq
Thus $ || X_{P_i} - \td{X}_{P_i} ||_2$ is forced to be small.
Hence, applying this method directly to group selection may lose power for strongly correlated features. Our numerical experiments confirm this.

\subsection{Dai-Barber's group knockoff filter}\label{limbar}

In \cite{Barber}, Dai and Barber introduced a group-wise knockoff filter, which is a generalization of the knockoff filter.
 Assume that the columns of $X$ can be divided into $k$ groups $\{ X_{G_1},X_{G_2},...,X_{G_k} \}$. The authors construct the group knockoff matrix  according to $\tilde{X}^T \tilde{X} = X^T X , \; \tilde{X}^T X =\Sigma - S, \; \Sigma = X^T X, $ where $S \succeq 0$ is group-block-diagonal, i.e. $S_{G_i,G_j}=0$ for any two distinct groups $i\ne j$. In the equi-correlated construction, $S = diag(S_1,S_2,...,S_k), \; S_i = \gamma \Sigma_{G_i,G_i} = \gamma X^{T}_{G_i} X_{G_i},i=1,2,...,k $. The constraint $S \preceq 2 \Sigma$ implies $\gamma \cdot diag(\Sigma_{G_1,G_1},\Sigma_{G_2,G_2},...,\Sigma_{G_k,G_k}) = S \preceq 2\Sigma$. In order to maximize the difference between $X$ and $\tilde{X}$, $\gamma $ is chosen as large as possible:
  $  \gamma = \min \{1,2\cdot \lambda_{\min} (D\Sigma D) \}$, where $D = diag(\Sigma^{-1/2}_{G_1,G_1},\Sigma^{-1/2}_{G_2,G_2},...,\Sigma^{-1/2}_{G_k,G_k}   ) .$
In the later numerical experiments, we will also use the SDP construction, which was not considered in \cite{Barber}, by solving
$\max \sum_i^k \gamma_i \  \textrm{subject to } \   S = diag(S_1,S_2,...,S_k) \preceq  2 \Sigma, \; S_i = \gamma_i \Sigma_{G_i,G_i} $ and $0\leq \gamma_i \leq 1$. This construction can be viewed as an extension of the SDP knockoff construction in \cite{Candes}. Due to the extra cost in solving the optimization problem, the SDP construction is more expensive than the equi-correlated construction.

The group-wise statistic introduced in \cite{Barber} can be obtained after the construction of the group knockoff matrix.
The construction above guarantees the group-wise exchangeability. Finally, group FDR control, i.e.
$ FDR_{group} \triangleq E\left[   \frac{\{ \# \{ i: \beta_{G_i} =0 , i \in \hat{S} \}} { ( | \hat{S}|\vee1  )  }   \right] \leq q$,
 is a result of the group-wise exchangeability. Here $\hat{S} = \{ j: W_j \geq T\}$ is the set of selected groups for a chosen group statistic $W_j$.

\subsection{Prototype knockoff filters}
\subsubsection{Prototype using the data}
In this subsection, we propose a prototype knockoff filter that takes advantage of the prototype features and improves the computational efficiency in the construction of knockoffs and statistical power of the Reid-Tibshirani prototype method. We assume that $X$ can be clustered into $k$ groups $X = (X_{C_1},X_{C_2},...,X_{C_k})$ in such a way that within-group correlation is relatively strong. We select the prototype features using a procedure similar to Reid-Tibshirani's prototype method. 

\paragraph{Step 1}
Split the data by rows into two parts
 $y = \left(
 \begin{array}{c}
  y^{(1)} \\
  y^{(2)} \\
 \end{array}
 \right)
 $ and $X=\left(  \begin{array}{c}  X^{(1)} \\  X^{(2)} \\  \end{array}  \right) $, where $y^{(1)} \in R^{ n_1}, y^{(2)} \in R^{n_2} , X^{(1)} \in R^{ n_1 \times p}, X^{(2)} \in R^{n_2 \times p}$ and  then choose a prototype $P_i$ for each group via the marginal correlation $P_i = \arg \max_{ j \in C_i} |X^{(1) T}_j y^{(1)}|$, using the first part of the data.

 \paragraph{Step 2} Let $Q = \{ 1,2,..,p \}  \backslash P$. The knockoff matrix $\wt{X^{(2)}} =\lt( \wt{X^{(2)}}_P , X^{(2)}_Q  \rt) $ obeys
\begin{align}
 &  \wt{X^{(2)}}^T_P           \wt{X^{(2) }}_P   = X^{(2) \ T}_P X^{(2)}_P  , \quad X^{(2)\ T }_PX^{(2)}_P  - X^{(2)\ T }_P\wt{X^{(2) }}_P = diag(s_P) , s_P \in R^{|P|}, \label{eq:con1} \\
&  \lt( \wt{X^{(2)}  }_{P_i} - X^{(2)}_{P_i} \rt)^T  X^{(2)}_{C^c_i} = 0 , \quad \textrm{ for } i = 1,2,..,k , \label{eq:excY}
 \end{align}
 where $X_{C^c_i}  = (X_{C_1},..,X_{C_{i-1}}   , X_{C_{i+1}}  ,..,X_{C_k})$. The construction of Reid-Tibshirani's prototype method \cite{Tib} requires $ \lt( \wt{X^{(2)}  }_{P_i} - X^{(2)}_{P_i} \rt)^T  X^{(2)}_{P_i^c}  = 0 $, which implies  $\lt( \wt{X^{(2)}  }_{P_i} - X^{(2)}_{P_i} \rt)^T X^{(2)}_{C_i \backslash P_i} = 0.$ If the within-group correlation is strong, $\wt{X^{(2)} }_{P_i}  $ is forced to be close to $X^{(2)}_{P_i}$ (see \eqref{eq:R-T}). For group selection, the within-group constraints are not necessary and we do not impose the constrains between $ \wt{X^{(2)}  }_{P_i}$ and $X^{(2)}_{C_i \backslash P_i } $ in \eqref{eq:excY}. Thus,  we can construct $\wt{X^{(2)}}_{P_i}$ that can maximize its difference with the original feature $X^{(2)}_{P_i}$.
\beq\label{eq:proj}
  \overline{   X^{(2)} }_{P_i}  = X^{(2)}_{P_i} -  X^{(2)}_{C_i^c}   ( X^{(2) \ T}_{C_i^c}  X^{(2)}_{C_i^c}  )^{-1} X^{(2) \ T}_{C_i^c}  X^{(2)}_{P_i} , \quad W_i  =   \overline{   X^{(2)} }_{P_i} / ||    \overline{   X^{(2)} }_{P_i} ||_2^2, \quad   W \in R^{ n_2 \times k}.
\eeq
We consider two constructions of $s_P$ : the equi-correlated construction $s_{P_i} =2 \lambda_{\min}( (W^TW)^{-1} ) \wedge  || X^{(2)}_{P_i}||_2^2 $ for all $i$  and  the SDP construction
\beq\label{eq:pfsdp}
\textrm{ maximize } \sum\nolimits_{ i=1}^k s_{P_i}   \quad \textrm{ subject to }  \quad diag(s_P)  \preceq 2 (W^T W)^{-1},     \quad 0 \leq s_{P_i} \leq   || X^{(2)}_{P_i}||_2^2.
\eeq
We then construct $\wt{X^{(2)} } _P = X^{(2)}_P  - W diag(s_P) + UC$, where $U \in R^{n_2 \times k}$ is an orthonormal matrix  with  $U^T X^{(2)} = 0 $
and $C^T C = 2 diag(s_P) - diag(s_P) W^T W diag(s_P)  , C \in R^{ k \times k}$ is the Cholesky decomposition. The existence of $U$ and  $C$ requires $n_2 \geq p+ k$ and $ diag(s_P)  \preceq 2 (W^T W)^{-1}$.

\paragraph{Step 3} Recycle the first part data and select features with recycling. We concatenate the original design matrix on the first part with the knockoff matrix on the second part
$
\td{X} =
\lt(
\begin{array}{c}
X^{(1)} \\
\wt{X^{(2)}}\\
\end{array}
\rt), \
\td{X}_P =
\lt(
\begin{array}{c}
X^{(1)}_P \\
\wt{X^{(2)}}_P \\
\end{array}
\rt).$
One can verify that \eqref{eq:con1} still holds true for  $( X, \td{X} )$ with the same $s_P$
\vspace{-0.07in}
\beq\label{eq:excX}
 \wt{X}^T_P           \wt{X}_P   = X^{  T}_P X_P  , \quad X^T_PX_P  - X^{ T }_P\wt{X}_P = diag(s_P) , s_P \in R^{|P|}.\\
 \vspace{-0.05in}
\eeq
Finally, run the knockoff filter on $y$ and  $\lt(  X_P, \td{X}_P \rt)$ as \eqref{cex5} to obtain $\hat{S}$, and select group $i$ if $P_i \in \hat{S}$.

The recycling procedure in knockoff filter was developed in \cite{Candes2}, in which the authors showed that it could raise the power substantially. We also observed this improvement in our numerical simulation.
In order to select the prototype in step 1 efficiently and retain a large difference between $X_{P_i}$ and $\td{X}_{P_i}$, we choose $n_1 = 0.2 n \vee 5 \max_i |C_i| $ and $n_2 = n - n_1$. The requirement $n_2 \geq p+ k$ implies $n \geq  1.25(p+k)$.

The main result of the prototype knockoff filter is that it controls the group false discovery rate.

\begin{thm}\label{ptfdr}
For any $q\in [0,1]$, the prototype knockoff filter using the knockoff and knockoff+ filter controls the group mFDR and group FDR respectively,
\[
mFDR_{group} \triangleq E \left[   \frac{ \# \{ i: \  \beta_{C_i} = 0, \ i \in \hat{S}   \}  }{ \# \{ i :\  i \in \hat{S} \} + q^{-1}     }  \right]
\leq q, \quad FDR_{group} \triangleq E \left[   \frac{ \# \{ i: \  \beta_{C_i} = 0, \ i \in \hat{S}   \}  }{ \# \{ i :\  i \in \hat{S} \} \vee 1     }  \right]
\leq q.
\]
\end{thm}

The result is a consequence of the Lemma below and the super-martingale argument \cite{Candes}.
 \begin{lem}(i.i.d signs for the null clusters). Let $\eta \in \{\pm 1 \}^k$ be a sign sequence independent of $W_P$, with $\eta_j = +1$ for all non-null clusters $j$ and $\eta_j \overset{i.i.d}{\sim} \{  \pm1\} $ for null clusters $j$. Conditional on $y^{(1)}$,  $ (W_{P_1},...,W_{P_k} ) \overset{d}{=} (W_{P_1} \eta_1, ..., W_{P_k} \eta_k  ) $.
 \end{lem}

\begin{proof} Recall that the features in the final knockoff screening process (step 3) is $\lt(  X_P, \td{X}_P \rt) $. The statistic $W_P$ can be written as $W_P = W( [ X_P, \td{X}_P ]^T [X_P, \td{X}_P],[X_P, \td{X}_P]^T y )$. Since the true model is $y = X\b + \e,  \ \e \sim N(0, \sigma^2 I_n )$, we have $y^{(i)} = X^{(i)} \b + \e^{(i)}, i = 1,2$, where $\e^{(1)}$ is the first $n_1$ components of $\e$ and $\e^{(2)}$ consists of the remaining components. In particular, the 
prototype set $P$ and $y^{(1)}$ are independent of $\e^{(2)}$ and conditional on $y^{(1)}$, the randomness of $W_P$ comes from $\e^{(2)}$ only. Following the analysis for the original knockoff filter \cite{Candes}, we just need to verify the exchangeability for the features and the response. The exchangeability for the features comes from \eqref{eq:excX}. Conditional on $y^{(1)}$, the exchangeability for the response is guaranteed by the invariance of $Var \lt( [X_P \ \td{X}_P ]_{swap(S^{\prime})}^T y \B| y^{(1)} \rt)$ for any $S^{\prime}$, which is a result of \eqref{eq:con1}, \eqref{eq:excX}, and  $E \lt[   ( X_{P_i} -  \td{X}_{P_i}  )^T y  \B|  y^{ (1) } \rt]  = 0$ for null clusters $i$ :
\[
\bal
&E\lt[   ( X_{P_i} -  \td{X}_{P_i}  )^T y \B| y^{(1)} \rt] = E\lt[   ( X^{(2)}_{P_i} -  \td{X}^{(2)}_{P_i}  )^T y^{(2)} \B| y^{(1)} \rt]  = E\lt[   ( X^{(2)}_{P_i} -  \td{X}^{(2)}_{P_i}  )^T (X^{(2)} \b + \e^{(2)} ) \B| y^{(1)} \rt]  \\
  =& E\lt[   ( X^{(2)}_{P_i} -  \td{X}^{(2)}_{P_i}  )^T X^{(2)} \b  \B| y^{(1)} \rt]  
 = ( X^{(2)}_{P_i} -  \td{X}^{(2)}_{P_i}  )^T X^{(2)} \b = ( X^{(2)}_{P_i} -  \td{X^{(2)}}_{P_i}  )^T( X_{C_i^c}^{(2)} \b_{C^c_i}  ) = 0.
\eal
\]
The first equality holds because $X_{P_i}$ agrees with $\td{X}_{P_i}$ in the first $n_1$ components. The third equality holds due to the fact that the first part of the data, $y^{(1)}$, is independent of $\e^{(2)}$, which generates the second part of the data, $y^{(2)}$, and has zero mean. The final equality follows from \eqref{eq:excY}.
\end{proof}

\subsubsection{A PCA prototype filter}
In this subsection, we propose a PCA prototype filter for group selection for some special cases. The PCA prototype filter works well under the following assumptions: (i) the within-group correlation is relatively strong; (ii) the features within each group are positively correlated; (iii) we know a priori that the signals in each group  have the same sign. Assume that $X$ can be clustered into k groups $ X = (X_{C_1}, X_{C_2} ,...X_{C_k} )$ as in the previous subsection. The PCA prototype filter follows steps similar to those described in the previous subsection. First of all, we calculate the first principal component $V_{i}$ for each group $X_{C_i} , 1 \leq i \leq k$. Secondly, we construct the knockoff matrix following a similar procedure
\beq\label{eq:proj2}
\bar{V}_i  =  V_i  -  X^{(2)}_{C_i^c}   ( X^{(2) \ T}_{C_i^c}  X^{(2)}_{C_i^c}  )^{-1} X^{(2) \ T}_{C_i^c}  V_i,  \quad  W_i = \bar{V}_i / || \bar{V}_i||_2^2,  \quad \td{V} = V - W diag(s_P)  + UC
\eeq
where $P = \{ 1,2,..,k\}$ and  $ s_P$ is obtained from the equi-correlated construction or the SDP \eqref{eq:pfsdp} with a slightly different constraint
$0 \leq s_i \leq 1$. $U$ is an orthonormal matrix with $U^TX = 0$ and $C$ is obtained by the same formula. Finally, we run the knockoff filter on $y$ and $[V, \ \td{V}]$ to obtain $\hat{S}$ and select group $i$ if $i \in \hat{S}$. Theorem \ref{ptfdr} holds true for the PCA prototype filter and the proof is similar.

\begin{remark}
In this paper, we focus on selecting one prototype for each group. The prototype knockoff filter can also be generalized to include a few prototypes for each group. 
\end{remark}

\subsubsection{Computing the projection}
For $X \in R^{n \times p} , n > p$, we design a recursive procedure to calculate $ \bar{ X }_{C_i}  = X_{C_i}  -  X_{C_i^c}   ( X^{ T}_{C_i^c}  X_{C_i^c}  )^{-1} X^{ T}_{C_i^c}  X_{C_i} $ for all $i$ with $O(np^2)$ flops. As a result, we can obtain $W$ in \eqref{eq:proj} with $O(np^2)$ flops. Similar result holds true for  \eqref{eq:proj2}. For simplicity, we assume $|C_i| = l$ for all i.

\noindent \textbf{  Algorithm : Projection $(X,k, l )$  }

\textbf{1. } If $ k= 1$, return $X$

\textbf{2.} Else : divide $X$ into two parts $a = \lfloor k / 2 \rfloor, \ X_1 = X(:, 1: al ) , \  X_2 =  X(: , (a l + 1) : kl ) $. Then compute the projection recursively : $W_1 = X_1 - X_2 (X_2^T X_2)^{-1} X_2^T X_1 , \  W_2 = X_2 - X_1 (X_1^T X_1)^{-1} X_1^T X_2$. $\bar{X_1} =\textbf{Projection}(W_1, a, l), \  \bar{X_2} =  \textbf{Projection}(W_2, k -  a, l)$, return $ (\bar{X_1}, \bar{X}_2) .$

For fixed $n, l$, let $a_k$ be the total flops. From the recursion, we have $a_1 \leq C,   a_{2^m} \leq 2 a_{ 2^{m-1}} + C 2 n (  2^{m-1}  l )^2 $ for some universal constant $C$. Simple calculation yields $a_{2^m} \leq C n( 2^m l )^2$. The monotonicity of $a_k$ implies $a_k \leq 4 C n(k l)^2 = O(np^2)$. Similar algorithm and analysis can be applied to $X$ with different size of groups.

\subsection{Numerical comparison study of different knockoff group selection methods}\label{RT-DB}
 In this subsection, we perform several numerical experiments to compare Reid-Tibshirani's prototype method,  the prototype knockoff filter, the PCA prototype filter and Dai-Barber's group knockoff filter. Throughout the section, the group size is $5$, the noise level is $1$, the nominal FDR is $20\%$ and we use the adaptive threshold $T_1$ defined in \eqref{cex5} and the knockoff+ selection procedure.

\paragraph{Simulated signals with no cancellation.}
We use the numerical example in \cite{Barber} to compare several methods. The design matrix $X \in R^{3000 \times 1000}$ is clustered into $200$ groups with $5$ features in each group. The rows of $X$ follow the $ N(0,\Sigma)$ distribution with columns normalized, where $\Sigma_{ii} =1$, $\Sigma_{ij} = \rho$ for $i \ne j$ in the same group and $\Sigma_{ij}= \gamma \cdot\rho$ for $i\ne j$ in a different group. We choose 20 groups ($l=20$) with one signal in each group. Specifically, we first choose $l$ groups $i_1, i_2,...,i_l$ randomly and then generate the signals $\beta_j$ at indices $j = C_{i_1,1}, C_{i_2,1},..,C_{i_l,1}$ (the first feature in the selected groups)
  $ \overset{i.i.d}{\sim} \{ \pm M\}$ and $\beta_j =0$ for other indices.  
  The signal amplitude $M$ is $3.5$.

  For Reid-Tibshirani's prototype method and our prototype knockoff filter, we choose $n_1 = 0.2 n  = 600$ and split the data $y$ and $X$ into two parts as described in both methods. Apply the first part of the data to obtain the prototype for each group and then construct the SDP knockoffs on $\td{X}^{(2)}$. For these two methods and the PCA prototype filter, we use the orthogonal matching pursuit (OMP) statistic \cite{Candes, OMP} and use the following short hand notations \textit{Reid-Tibshirani knockoff+}, \textit{Prototype knockoff+} and \textit{PCA knockoff+} in the following Figures.
  For the group knockoff filter, the first method is to construct the equi-correlated group knockoffs and then apply the group Lasso path (GLP) statistic, which is the method discussed in \cite{Barber}. We also consider two other group selection methods based on the SDP group knockoffs and the OMP statistic. After constructing the SDP group knockoffs $\td{X}^g$, we extract the first principal component of each group $X_{C_i}, \td{X}^g_{C_i},i =1,2,..,k$, which form $V, \td{V} \in R^{n \times k}$. We then run the knockoff filter on $y$ and $[V, \  \td{V}]$ with the OMP statistic. This method is an analog of the PCA prototype filter with a different construction on $\td{V}$. Meanwhile, it is equivalent to the group knockoff with a special group knockoff statistic and thus the group FDR control follows from \cite{Barber}.
  The other method uses a group version of the OMP statistic defined below. After selecting group $j_t, t \geq 0 $, we define $r_t$ to be the residual of the least square regression of $y$ onto  $\{ X_{C_{j_1}  } ,X_{C_{j_2}}  ,...,X_{C_{j_t}} \}$ ($r_0 = y$) and choose group $j_{t+1}$ via
$ j_{t+1} = \arg \max_j  ||  r_t^T X_{C_j} ||_2.$
Let  $Z_j , \td{Z}_j $ be the reversed order when the variable $X_{C_j}$ or $\td{X}_{C_j}$ enters the model, e.g. $Z_j  \ (\td{Z}_j)= 2k$ if  $X_{C_j} \ (\td{X}_{C_j})$ enters the first, where $k$ is the number of groups. We then define the group OMP statistic as $W_j = \max ( Z_j ,\  \td{Z}_j) \cdot sign(Z_j  - \td{Z}_j), j =1,2,..,k.$ We use the following short hand notations \textit{Group knockoff+ GLP}, \textit{PCA Group knockoff+} and \textit{Group knockoff+ OMP} to stand for these three group knockoff methods in the following Figures.

  To study the effect of within-group correlation, we fix the between-group correlation factor $\gamma =0$ and vary $\rho = 0, 0.1,...,0.9$. To study the effect of between-group correlation, we choose the within-group correlation factor $\rho = 0.5$ and $\rho=0.9$ while varying $\gamma =0,0.1,...,0.9$.  Each experiment is repeated 100 times. The group Lasso path is calculated via the \textit{gglasso} package \cite{gglasso} in $R$ with number of $\lambda$ equal to 1000.
  \begin{figure}[h]
   \centering
\includegraphics[width =\textwidth ]{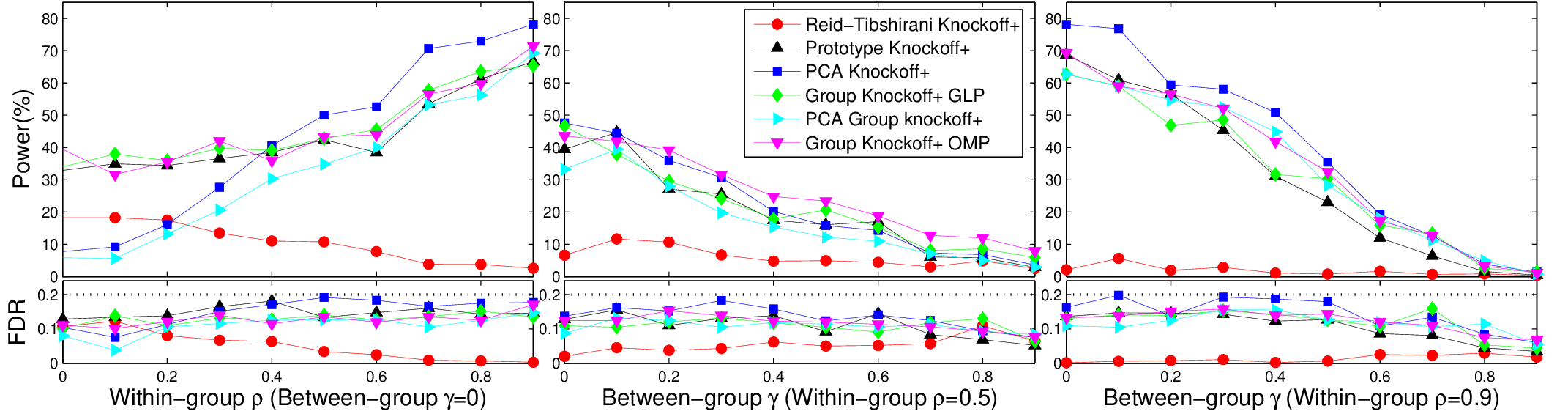}
 \caption{Testing Reid-Tibshirani's prototype method, the prototype knockoff filter, the PCA prototype filter and the group knockoff filter with varying  within-group correlation or between-group correlation. Here, we use the knockoff+ selection procedure. }
 \label{gp1}
 \end{figure}

For each design matrix $X$ and its knockoff $\td{X}$,  we consider the measurements  :  average of $ || X_{P_i} - \td{X}_{P_i}||_2^2 / 2  = s_{P_i} $ over all prototype features for the Reid-Tibshirani's prototype knockoff, the prototype knockoff, the PCA prototype knockoff, $ || X_i - \td{X}^g_i ||_2^2 / 2 =  2 \lam_{\min}(D\Sigma D)  $ uniformly for the equi-correlated group knockoff and average of $||X_i  - \td{X}^g_i||^2_2/2$ (i.e. average $\gamma_i$) over all features for the SDP group knockoff. In three experiments, the mean values of this average  $(  \bar{s}^{RT}_{P_i} , \bar{s}^{Prototype}_{P_i} , \bar{s}^{PCA}_{P_i}, \lam^{group}_{\min},\bar{\gamma}_i^g )$ (10 design matrices in each experiment) are $(0.14, 0.50, 0.79, 0.36, 0.38)$, $( 0.13, 0.39,  0.49, 0.29, 0.31) $ and $ ( 0.03  ,  0.30 ,   0.44,    0.26, 0.28)$.  We have performed numerical experiments for a general class of design matrices and found that the PCA prototype knockoff has the largest average difference and the average difference $||X_{P_i} - \td{X}_{P_i}||_2^2$ of the prototype knockoff is slightly larger than that of the group knockoff. These two prototype methods overcome the problem of strong within-group correlation and improve the power of the Reid-Tibshirani's prototype method significantly.

The performance of the prototype knockoff filter is comparable to that of the group knockoff with the GLP and the OMP statistics. When the within-group correlation is strong and the between-group correlation factor $\gamma$ is relatively small (the left and the right subfigures), the first principal component captures most of the information and the PCA prototype filter outperforms the group knockoff filter with the GLP or the OMP statistic due to the larger difference between the prototype feature and its knockoff. In the middle subfigure where the within-group correlation is $0.5$, the PCA prototype filter offers slightly less power than that of  the group knockoff filter with the OMP statistic due to the relatively weak within-group correlation. Comparing the performance of the PCA knockoff+ and the PCA group knockoff+ in Figure \ref{gp1}, we find that the PCA knockoff+ consistently offers more power than that of the PCA group knockoff+,
which justifies that the knockoff construction that we propose indeed increases the power. The advantage of larger $s_{P_i}$ in the PCA prototype knockoff can be exploited if we know a priori that the signals in each correlated group have the same sign. We have also constructed the equi-correlated group knockoff and then applied the group OMP statistic. It offers power similar to that of the Group knockoff+ OMP,  see  Figure \ref{gp1}. For the later numerical experiments, we only focus on the power of different methods since these methods are guaranteed to control FDR.

\paragraph{Simulated signals with cancellation.}
In the second example, we first generate $X \sim N(0, \Sigma)$ with within-group correlation $\rho =0 , 0.1,..,0.9$ and between-group correlation $0$ as in the previous example. We then generate signals with cancellations in the first $20$ correlated groups. We consider two settings of signal amplitude: (a) partial cancellation : $\b_{C_i} = (a_i, -0.5a_i, 0,0..,0) $; (b) complete cancellation :  $\b_{C_i} = (a_i, -a_i, 0,0..,0) $,  where $a_j  \overset{i.i.d}{\sim} \{ \pm 5\}, i = 1,2,..,20$. Other settings remain the same as in the previous example. We construct the SDP knockoff for the prototype knockoff filter and the SDP group knockoff and focus on the prototype knockoff filter, the PCA prototype knockoff filter and the group knockoff  filter with the OMP statistic.

  \begin{figure}[h]
   \centering
 \includegraphics[width =\textwidth ]{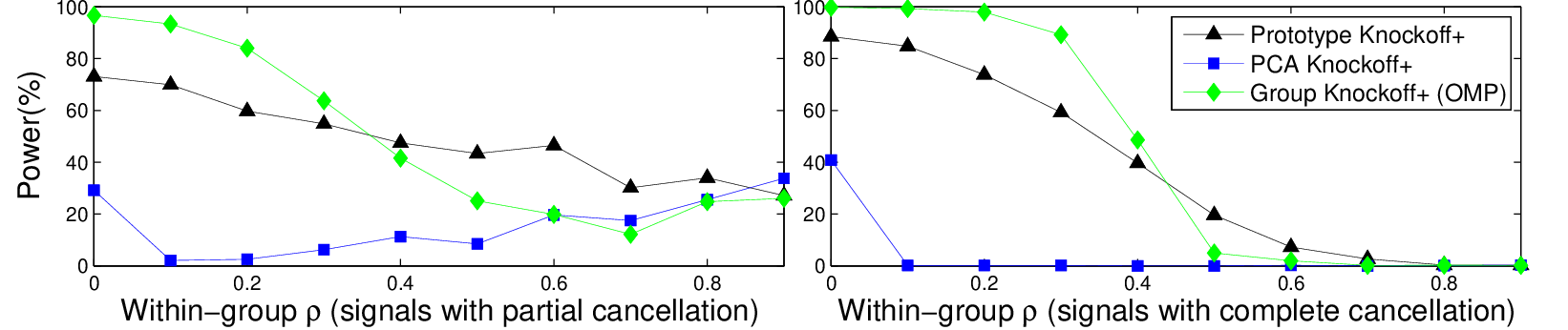}
\caption{Comparing the powers of the prototype knockoff filter, the PCA prototype knockoff filter, and the group knockoff filter using signals with cancellation. }
 \label{gp_cancel}
 \end{figure}

In Figure \ref{gp_cancel}, when the within-group correlation is weak, the group knockoff filter captures more signals and offers more power than that of the prototype knockoff filter and the PCA prototype knockoff filter. For large $\rho$, two signals almost merge into one signal and the prototype knockoff filter slightly outperforms the group knockoff. The PCA prototype filter loses considerable power in this case since the projection of the signals within the group onto the first principal component direction suffers from the cancellation. Without any prior knowledge about the signal, the first principal component may not be a good prototype and we recommend to construct a data-dependent prototype via marginal correlation for the prototype knockoff filter.

In the following numerical experiments,  without specification, we use the SDP construction for prototype knockoff filters and the group knockoff.
We focus on the prototype knockoff filter, the PCA prototype knockoff filter, the group knockoff filter with the associated OMP statistic, and the group knockoff filter with the GLP statistic.

\paragraph{Group features spanned by one hidden factor. }
For $j=1,2,..,100$, we generate group features $X_{C_j} (\xi) = (1,1,..,1)  \cdot \cos j\xi  \in R^5$ with group size 5.  In total, we have $p = 500$ features. We then generate $n =3p =  1500$ i.i.d realizations of $\xi$, $\xi \sim Unif[0, 2\pi] $, assemble $X_{C_j} (\xi)$ by rows and normalize the columns to obtain  $X_{C_j} \in R^{n \times 5} $ and  $X_0 = (X_{C_1},   X_{C_2}, ..,X_{C_{100}} ) \in R^{n \times p}$. To avoid linear dependence, we perturb $X_0$ by some white noise: $X_0+ \sigma \cdot \td{G}$, where $\td{G}\in R^{n\times p}$ is obtained by normalizing the columns of  $G\in R^{n\times p}$, $G_{ij} \overset{i.i.d}{\sim} N(0,1)$. We then normalize the columns of  $X_0+ \sigma \cdot \td{G}$ to obtain the design matrix $X$ and the modified group features matrix $X_{C_j}$.
We select the last 15 groups (group with high frequency features), and generate the signal $\beta_{C_j}$ in the selected group via one of the settings of signal amplitude (a) $\beta_{C_j} = (a_j,0,0,0,0 )^T, a_j \overset{i.i.d}{\sim}  \{ \pm 3.5 \}$ ; (b) $ \beta_{C_j} = (a_j, - a_j,0,0,0 )^T,  a_j \overset{i.i.d}{\sim}  \{ \pm 6 \}$; (c) $\beta_{C_j} = (a_j, a_j, -a_j,0,0 )^T, a_j \overset{i.i.d.}{\sim} \{ \pm 3.5 \}$. We then generate $y$ as follows : $y = \sum_{j= 86}^{100}  X_{C_j} \beta_{C_j} + \e, \ \e \sim N(0, I_n)$. Note that the noise level is 1 and in setting (b), (c), there are cancellations in the signals.


By definition of $X_{C_j}$, the within-group correlation mainly depends on $\sigma$. Smaller $\sigma$ results in larger within-group correlation. We vary $\sigma = 1,0.9,..,0.1$ and repeat 100 times for each experiment to compare the performance of several methods. The results are plotted in Figure \ref{gp_pca1}.

  \begin{figure}[h]
   \centering
 \includegraphics[width =\textwidth ]{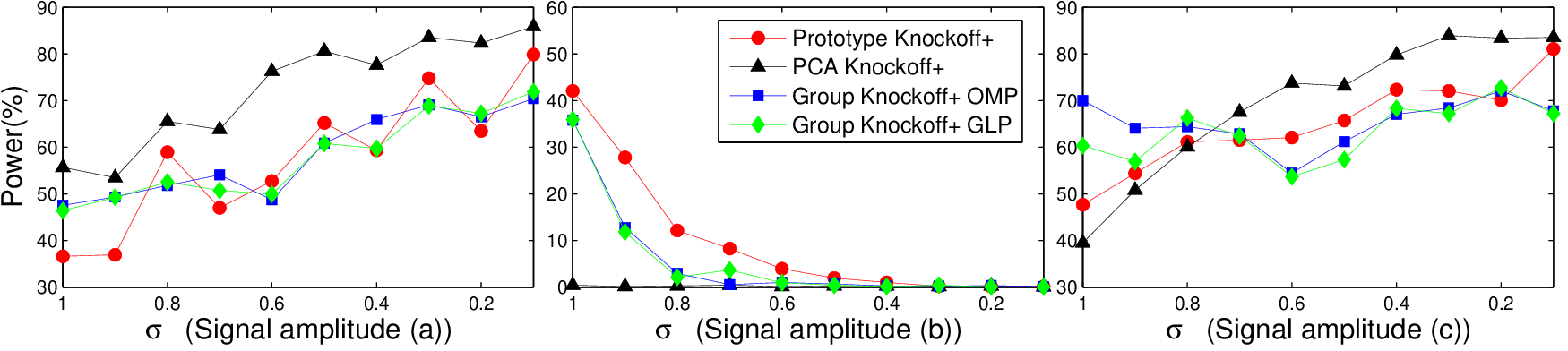}
 \caption{Comparing the powers of the prototype knockoff filter, the PCA prototype knockoff filter and the group knockoff filter by varying within-group correlation and signal amplitude. Here, the group features are generated by one hidden factor. }
 \label{gp_pca1}
 \end{figure}

In the case of signal setting (a) and of signal setting (c) with small $\sigma$ ($\sigma \leq 0.7$) , the PCA prototype filter offers significantly more power than that of the group knockoff filter with both statistics. The performance of the prototype knockoff filter is comparable to that of the group knockoff filter and it offers more power in the case (b).
The equi-correlated group knockoff with both statistics (not plotted) offer power similar to that of the SDP group knockoff.

\paragraph{Group features spanned by two hidden factors. }
We first generate $k=100$ low frequencies $\om_{L,j}  \overset{i.i.d}{\sim} Unif[0, 10] $ and $k$ high frequencies $\om_{H,j} \overset{i.i.d}{\sim} Unif[100,200]$. The feature in group $j$ is a convex combination of $\cos \om_{L,j} \xi $ and $\cos \om_{H,j} \xi$:
\[
X_{C_j}(\xi,  \tau_j) =  (  \tau_{j, 1}  \cos \om_{L,j} \xi + (1 - \tau_{j,1}) \cos \om_{H,j} \xi \ , ..., \
\tau_{j, 5}  \cos \om_{L,j} \xi + (1 - \tau_{j,5}) \cos \om_{H,j} \xi  ) \in R^5
\]
where $\xi \sim Unif[0, 100]$ and $\tau_{j,i} \overset{i.i.d}{\sim} Unif[0, \tau]$ for a parameter $\tau \in [0, 1]$. By definition, the within-group correlation mainly depends on $\tau$. If $\tau =0$, the features in group $j$ are spanned by one factor $\cos \om_{H,j} \xi $. Applying a procedure similar to the one used in the previous example, we can generate $X_0 \in R^{1500 \times 500}$. To avoid linear dependence, we perturb $X_0$ by some white noise and normalize $X_0 + 0.3 \cdot \td{G}$ to obtain the design matrix $X$. Other settings, including 3 signal amplitudes and the sparsity (15 selected groups), remain the same as in the previous example. Due to the randomness in generating $\tau_j$ and the within-group correlation of $X$, we generate 5 $X_{\tau} \in R^{n\times p}$ for each $\tau = 1,0.9,..,0.1$ and consider the average results. For each $X_{\tau}$, we construct simulated data and repeat the experiment 100 times to obtain an average power and FDR. We plot the power averaged over  five $X_{\tau}$ in Figure \ref{gp_pca2}.

 \begin{figure}[h]
   \centering
 \includegraphics[width =\textwidth ]{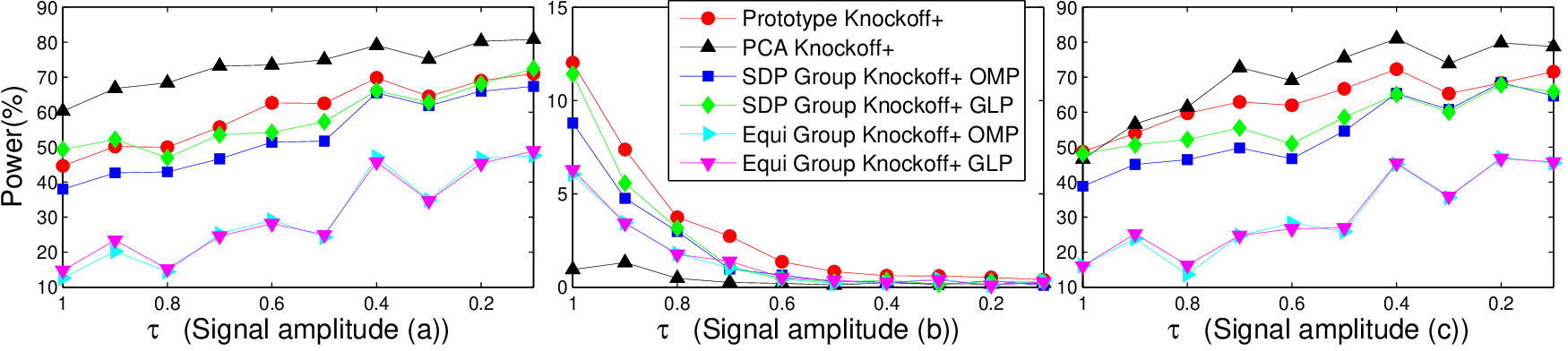}
 \caption{Comparing the powers of the prototype knockoff filter, the PCA prototype knockoff filter and the group knockoff filter by varying within-group correlation and signal amplitude. Here, the group features are generated by two hidden factors.}
 \label{gp_pca2}
 \end{figure}

In the case of (a) and (c), the PCA prototype filter outperforms the group knockoff filter and the prototype filter offers more power than that of the group knockoff filter almost in all cases. We can see that the equi-correlated group knockoff with the GLP or the OMP statistics (Equi- Group knockoff+ GLP, OMP in Figure \ref{gp_pca2}) only offers about $50\%$ power of the SDP group knockoff. In this example, the average of $||X_i - \td{X}^g_i||^2_2 / 2$ over all features of the equi-correlated group knockoff is $0.1402$
and is much smaller than that of the SDP group knockoff, which is $ 0.3581$. This explains the loss of considerable power.

From the last two examples, we observe that when the group features are spanned by one or a few hidden factors and the signals within each group are not canceled completely, the PCA prototype filter could offer more power than  that of  the group knockoff filter. If the within-group correlation is relatively strong and the signals are canceled completely as in the case (b) with $\sigma \leq 0.8$ in Figure \ref{gp_pca1}, \ref{gp_pca2}, the information is lost in the measurement $y$ and it is challenging to perform group selection in this case.

\paragraph{Computational efficiency}
The computational cost of the prototype knockoff filter or the group knockoff mainly consists of the knockoff construction and the feature selection process. We can apply the equi-correlated or the SDP construction. In the equi-correlated construction, computing the smallest eigenvalue $\lam_{\min} ( (W^TW)^{-1} )$ or $\lambda_{\min} ( D\Sigma D )$ is relatively cheap and both methods construct knockoffs in $O(np^2)$ flops (the computation of $X^TX$ is one of the bottlenecks). In the SDP construction of the prototype knockoff,  \eqref{eq:pfsdp} is a  $k-$dimensional SDP ($k$ is the number of groups). If $k \ll p $, the SDP construction of the prototype knockoff can be solved efficiently by exploiting its special structure \cite{convex}. We remark that the construction of the SDP knockoffs in \cite{Tib} solves a $p-$dimensional problem of the same type.  By default, we use the SDP construction in the prototype knockoff filter. The performance of the equi-correlated group knockoff depends on the amplitude of $\lambda_{\min} ( D\Sigma D ) $, which is exactly the average $|| X_i -\td{X}_i||_2^2 / 4$ and could be much smaller than the corresponding average of the SDP group knockoff. In this case, the equi-correlated group knockoff could lose significant power as we have demonstrated in the last example and one may have to construct the SDP group knockoff, which is more expensive.


The prototype knockoff filter constructs the knockoff statistic using $y, ( X_P, \tilde{X}_P) \in R^{n \times 2k}$. The computational cost of many useful statistics, e.g. the Lasso path, the OMP, the ridge regression and the Lasso statistics, is $O(n k^2)$, is relatively small compared with the cost in the knockoff construction.


\section{Some observations of the knockoff filter}\label{sec_ana}
\subsection{Alternating sign effect}
 In this section, we will perform some analysis for the knockoff filter and illustrate a potential challenge that we may encounter for some path method statistics, including the Lasso path (LP) statistic and the forward selection statistic \cite{Candes, FS} for certain implementation procedure to update the residual. After performing $l$ steps in one of the path methods (or at $\lambda$ for the Lasso path), we denote by $E$ the set of features that have entered the model. 
Assume that $X_j, \tilde{X}_j \notin E$ at the $l$th step,  but at the next step either $X_j$ or $\tilde{X}_j$ will enter the model. After $l$ steps, the residual is
$r_l = y - X_E\hat{\beta}_E = X\b - X_E \hat{\b}_E + \e$.
Since $X_j ,\tilde{X}_j \notin E$, we have $X_j ^TX_i = \tilde{X}_j X_i , \; \forall X_i\in E$.
 The same equality holds for $ \tilde{X}_i$. For $X_j,\tilde{X}_j$, their marginal correlation with $r_l$ determines which one of these two features will enter into the model first at the $(l+1)$st step:
 \[
 \begin{aligned}
 & X_j^T r_l = X_j^T (X\beta - X_E\hat{\beta}_E) + X_j^T\epsilon ,  \quad   \tilde{X}_j^T r_l = \tilde{X}_j^T (X\beta - X_E\hat{\beta}_E)  + \tilde{X}_j^T\epsilon ,\\
 &  (  X_j -\tilde{X}_j)^T r_l =  (X_j-\tilde{X}_j)^Ty =  s_j \beta_j + (X_j -\tilde{X}_j)^T \epsilon .\\
 \end{aligned}
 \]
Assume that the noise level is relatively small. If

\beq\label{eq:alt_assume}
 sign(\beta_j) \ne sign(X_j ^T(X\beta -X_E \hat{\beta}_E))   \textrm{ and } \  |X_j ^T(X\beta -X_E\hat{\beta}_E)| > |s_j \beta_j|,
 \eeq
  then  $\tilde{X}_j$ will enter into the model at the $(l+1)$th step since
\[
   \begin{aligned}
    &   |\tilde{X}_j ^T r_l | -|X_j ^T r_l | \approx  | X_j ^T(X\beta -X_E \hat{\beta}_E) - s_j \beta_j | -| X_j ^T(X\beta -X_E \hat{\beta}_E)| \\
   =& \; sign\left( X_j ^T(X\beta -X_E \hat{\beta}_E)\right) \left[ \left(  X_j ^T(X\beta -X_E \hat{\beta}_E) - s_j \beta_j \right) -  \left(X_j ^T(X\beta -X_E \hat{\beta}_E) \right)\right]  = |s_j\beta_j|  >0.\\
   \end{aligned}
\]
This may reduce the power of the knockoff filter.

To understand under what condition the assumption \eqref{eq:alt_assume} could be satisfied, we simply replace $X\beta -X_E \hat{\beta}_E $ by $y$.
Then this assumption can be reformulated as $sign(X_j^Ty ) \ne sign(\b_j)$ and $| X_j^T y|  > |s_j \beta_j| $. Suppose that the correlation between the features is relatively strong. The direction $X_j$ can capture more signals from $y$ and thus it is likely that the marginal correlation $|X_j^T y|$ is larger than $| \b_j|  > |s_j \b_j|$. This could occur if some features are positively (negatively) correlated but their contribution to the response $y$ has the opposite (same) sign, e.g. $X_j^T X_k > 0,  \b_j > 0, \b_k < 0$. 
We call this mechanism that could lead to \eqref{eq:alt_assume} the \textit{alternating sign effect}.

 \paragraph{A numerical example}
We generate the design matrix $X \in R^{900 \times 300}  \sim N(0, \Sigma)$, where $\Sigma_{ij} =  0.9^{|i-j| }$, and construct its SDP knockoffs. Since some columns of $X$ are strongly correlated, the knockoff factor $s \in R^{300}$  (defined in \eqref{revko1}) obtained in the SDP construction can be very small. Hence, some $X_i$ is very close to its knockoff $\td{X}_i$,  which may lead to the degeneracy of the augmented design matrix $[X  \ \td{X}]$. Denote $S \teq \{  i: s_i  \leq 0.01\}$ and $L \teq  \{  i: s_i   >  0.01\}$. We pick $k = 30$ features $i_1, i_2,..,i_k$ randomly in the set $L$ and then generate the signal amplitude $\beta_{i_j}   =  \f{0.75}{ s_{i_j} }$. We then pick half of the signals $\beta_{i_j}$ randomly and change their signs. By construction, $50\%$ of the signals are positive. We construct $y = \sum_{j=1}^k X_{i_j} \beta_{i_j} + \e, \ \e \sim N(0, \sigma^2  I_n)$ and run the knockoff+ filter 
with the LP, the forward selection and the Lasso statistics with tuning parameter $\lam =\sigma$ on $y$, $[X_L, \td{X}_L, X_S]$ to obtain the statistic $W_L$. We set $W_S = 0$. For the forward selection statistic, initializing $r_0 = y$,
we iteratively choose $X_{i_l}$ ($l\geq 1$) via $i_l = \arg \max_j |\la  r_{l-1} , X_j \ra |$. According to \cite{Candes}, we can apply two different procedures to update the residual at step $l$. 
In the first procedure, the residual $r_l$ is simply updated by eliminating the effect of the selected variable $X_{i_l}$ from the previous residual $r_{l-1}$, while in the second procedure $r_l$ is updated by eliminating the effect of all selected variables from $y$, i.e.
\beq\label{FS_residual}
r^{(1)}_l  = r^{(1)}_{l-1} - \la r^{(1)}_{l-1}, X_{i_l}  \ra X_{i_l}, \quad r^{(2)}_l = y - P_l y
\eeq
where $P_l$ is the orthogonal projector onto the space spanned by the $l$ selected variables. For the forward selection statistic with these two procedures, we use short hand notations \textit{FS, OMP} since the latter uses orthogonal matching pursuit. 
The FS and OMP statistics can be computed by the knockoff package \cite{Candes} directly. We also apply the Benjamini-Hochberg (BHq) procedure \cite{BHq1} that first calculates the least-squares estimate $\hat{\b}^{LS} = (X^T X)^{-1} X^T y$ and set $Z_j = \hat{\b}_j^{LS} /  ( \sigma  \sqrt{ (G^{-1} )_{jj}  })$ to yield the z-scores, where $G = X^TX$ is the Gram matrix. Note that marginally $Z_j \sim N(0, 1)$. Variables are then selected by the threshold $T \teq  \min \{   t : 300 \cdot P(  | N(0,1) | \geq t  ) / \# \{ j : |Z_j| \geq t \} \leq 20\%
\} $. 

Due to the randomness in generating the signal $\b$, we generate 5 signals for each noise level $\sigma = 1.4 , 1.2,..,0.2$ and consider the average and the extreme results. For each signal and $\sigma$, we repeat the experiment 100 times to obtain an average power and FDR. For each $\sigma$, we calculate the power, the FDR averaged over 5 cases with different signals, the minimal power and its associated FDR over 5 cases. The results are plotted in Figure \ref{alter}.

\begin{remark}
In the selection procedure, we effectively turn off the knockoff $\td{X}_i$ for $i \in S$, i.e. $s_i$ is small. If we run the knockoff process on $y$ and the whole augmented design matrix $M$, the degeneracy of $M$ can lead to significant numerical instability and the loss of FDR control. Note that $W_S =0$ implies that the features in set $S$ will not be selected by the knockoff filter.  Since the signals are from $L$ by the design of the response $y$, setting $W_S=0$ will not lead to the loss of power.

 In our computation, we use the knockoff package in Matlab to calculate the Lasso path and the FS statistics. This package uses the \textit{glmnet} package in Matlab \cite{glmnet} to solve the Lasso problem and we also use the \textit{glmnet} package to obtain the Lasso statistic. In an earlier example (not included in the current paper) where the features are strongly correlated, we obtained a somewhat unexpected result, i.e. the LP statistic fails to control FDR. To gain some understanding what went wrong, we found that the numerical solution of the Lasso problem $(\hat{\b}, \td{\b}) = \arg \min_{( \hat{b}, \td{b}) } \f{1}{2}  || y-  X \hat{b} -\td{X} \td{b}  ||_2^2 + \lambda  || (\hat{b}, \td{b}) ||_1$ is significantly different from the numerical solution of $(\hat{\b}, \td{\b}) = \arg \min_{ ( \hat{b}, \td{b}) } \f{1}{2}  || y-  \td{X}\td{b}  - X\hat{b}  ||_2^2 + \lambda  || (\hat{b}, \td{b}) ||_1$, which is the same Lasso problem except that we have swapped the order of the input variables $(X, \td{X})$. Therefore, when the features of the design matrix $[X \td{X}]$ are strongly correlated, the result of the Lasso related methods obtained by the \textit{glmnet} package could suffer from a large numerical error, which may lead to the loss of FDR control. 
The problem may be addressed by significantly increasing the precision of the solver, e.g. using the option \textit{thresh} $ = 10^{-12} $. To avoid this problem, we simply exclude $\td{X}_S$ when we run the knockoff selection procedure. 
\end{remark}

\begin{figure}[h]
 \centering
   \includegraphics[width =\textwidth ]{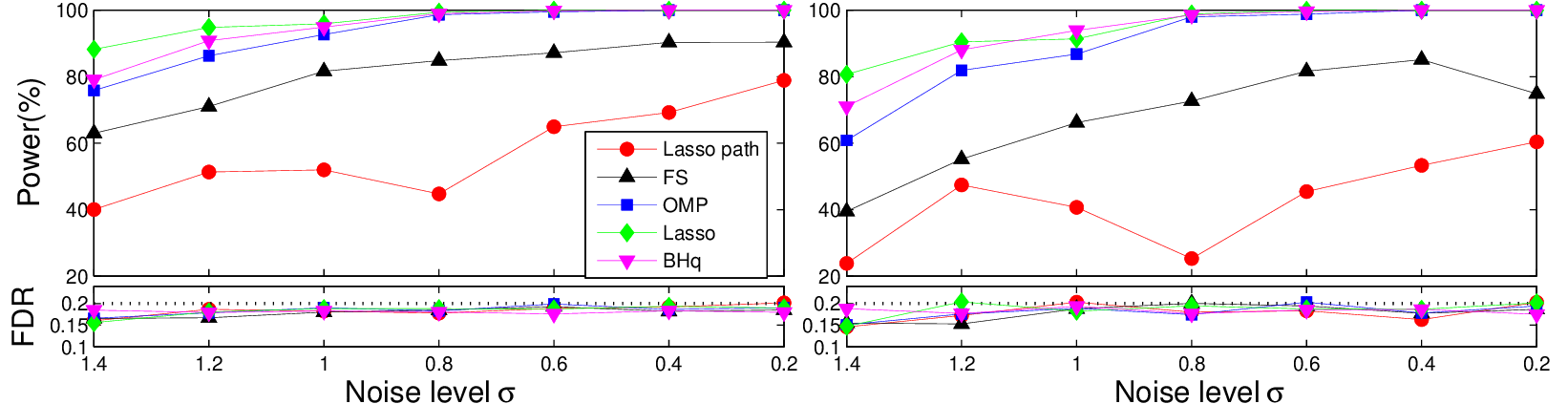}
   \caption{Testing different knockoff statistics using example with positive correlated features and signals with different sign. In the left figure, for each $\sigma$, the power and the FDR are averaged over 5 cases. In the right figure, the minimal power over 5 cases and its associated FDR are presented.
 }
   \label{alter}
 \end{figure}

 As the noise level $\sigma$ decreases, the signal becomes relatively strong and we expect that the power increases. For $\sigma  \leq 1$, the power of the BHq method is over $90\%$ and its FDR is under control, which is a good indicator that the signal is strong enough. For $\sigma =0.2, 0.4$,  the minimal power of the knockoff with the LP ($53 \%, 60\%$) and the FS ($85\%, 75\%$) statistics are significantly less than that of the BHq  method ($100\%, 100\%$), which suggests that these two statistics suffer from 
the alternating sign effect and lose significant power in this example. The power of the knockoff with the OMP or the Lasso statistic is comparable to that of the BHq method. 
 A possible explanation of the robust performance of the OMP and the Lasso statistics in this example is that the OMP statistic can detect strong signals and then eliminate the effect of all selected variables before finding the next strong signal, and the Lasso statistic jointly estimates the effect of all regressors. Compared with the OMP, the forward selection with the first procedure in \eqref{FS_residual} (FS) fails to eliminate the effect of all selected variables, which may lead to the loss of considerable power.
 In strongly correlated cases, the OMP or the Lasso statistic is less likely to suffer from this effect. We have also implemented a similar test where the features are only weakly correlated:  $X  \sim N(0, \Sigma), \Sigma_{ij} = 0.5^{|i-j|}$. The power of the knockoff with the LP and the FS statistics is comparable or more than that of the BHq method. Thus it is unlikely that these statistics would lose power for weakly correlated features.

 \subsection{Extension of the knockoff sufficiency property}\label{ext-sta}

 Let $U\in R^{n \times (n-2p)}$ be an orthonormal matrix such that $[ X \    \tilde{X}]^T U = 0$ and  $ [X \   \tilde{X}  \  U]$ admits a basis of $R^{n}$.
 (\ref{revko1}) implies $(X+\tilde{X})^T ( X -\tilde{X}) = X^T X -\tilde{X}^T \tilde{X} =0$. Hence, $R^n$ can be decomposed as follows
 \[
 R^{n} = span( X+\tilde{X} ) \oplus span(X-\tilde{X}) \oplus span(U) .
 \]
 Our key observation is that swapping each pair of the original $X_j$ and its knockoff $\tilde{X}_j$ does not modify these spaces: $span( X+\tilde{X} )$, $span(X-\tilde{X})$ and $span(U)$. Therefore, the probability distributions of the projections of the response $y$ onto these spaces respectively are independent and invariant after swapping arbitrary pair $X_j, \tilde{X}_j$. Inspired by this observation, we can generalize the sufficiency property of knockoff statistic \cite{Candes} which states that the statistic $W$ depends only on the Gram matrix $[X \tilde{X}]^T[X \tilde{X}]$ and the feature-response product $[ X \tilde{X}]^T y$.
 \begin{definition}[Generalized Sufficiency Property]
 The statistic $W$ is said to obey the generalized sufficiency property if $W$ depends only on the Gram matrix $[X \tilde{X}]^T[X \tilde{X}]$ and the feature-response $[X\  \tilde{X}\ U]^T y$; that is, we can write
 $ W = f( [X \tilde{X}]^T[X \tilde{X}], [X\  \tilde{X}\ U]^T y   )$
 for some $f : S_{2p}^+ \times R^n \rightarrow R^p$ and an orthonormal matrix $U \in R^{n \times (n-2p) }$ that satisfies $U^T [X \tilde{X}] =0$.
 \end{definition}

 The definition of the antisymmetry property remains the same: swapping $X_j $ and $\tilde{X}_j$ has the same effect as changing the sign of $W$, i.e.
 \[
 W_j( [X \tilde{X}]_{swap(\hat{S})}, U, y   ) = W_j( [X \tilde{X}], U, y   )\cdot
 \begin{cases}
 + 1 & j\notin \hat{S},\\
 -1 & j\in \hat{S}, \\
 \end{cases}
 \]
 for any $\hat{S} \subset \{1,2,..,p\}$.  For any knockoff matrix $\tilde{X}$ and the associated statistic $W$ that satisfies the above definition, we call $W$ the generalized knockoff statistic.
 Following the proof of Lemma 1, 2 and 3 in \cite{Candes}, one can verify the pairwise exchangeability for the features and the response.
 \begin{lem}\label{exgkf}
 For any generalized knockoff statistic $W$ and a subset $\hat{S}$ of nulls, we have
 \[
 W_{swap(\hat{S})}  = f( [X \tilde{X}]_{swap(\hat{S})} ^T [X \tilde{X}]_{swap(\hat{S})}, [ \  [X   \tilde{X}]_{swap(\hat{S})} \ U  ]^T y ) \overset{d}{ =}  f( [X \tilde{X}]^T[X \tilde{X}], [X\  \tilde{X}\ U]^T y   ) = W.
\]
 \end{lem}

Moreover, we can show that the ``i.i.d. signs for the nulls'' property still holds true for the generalized knockoff statistic.
 \begin{lem}(i.i.d signs for the nulls). Let $\eta \in \{\pm 1 \}^p$ be a sign sequence independent of $W$, with $\eta_j = +1$ for all nonnull $j$ and $\eta_j \overset{i.i.d}{\sim} \{  \pm1\} $ for null $j$. Then $  (W_1,...,W_p ) \overset{d}{=} (W_1 \eta_1, ..., W_p \eta_p  ) $.
 \end{lem}
Based on these lemmas,  we can apply the same super-martingale as in \cite{Candes} to establish rigorous FDR control.

 \paragraph{Estimate of the noise level}
A natural estimate of the noise level is $\hat{\sigma} = || y - X \hat{\b}^{ls} - \td{X} \td{\b}^{ls} ||_2 / \sqrt{n-2p}$ provided $ n > 2p$, where $\hat{\b}^{ls} ,\td{\b}^{ls}$ are the least squares coefficients. Let  $U \in R^{n \times (n-2p) }$ be an orthonormal matrix such that $ U^T [X \tilde{X}] =0$. It is straightforward to show that
\[ \hat{\sigma} =|| y - X \hat{\b}^{ls} - \td{X} \td{\b}^{ls} ||_2 / \sqrt{n-2p} =  || U^T y||_2 /\sqrt{n-2p} . \]
Note that $\hat{\sigma}$ depends on $U^Ty$ only.  As an application of the generalized knockoff statistic, we can incorporate this noise estimate into the knockoff statistic without violating the FDR control.

\section{Concluding Remarks}

In this paper, we proposed a prototype knockoff filter for group selection by extending Reid-Tibshirani's prototype method. Our prototype knockoff filter improves the computational efficiency and statistical power of the Reid-Tibshirani prototype method when it is applied for group selection. We demonstrated that when the group features are spanned by one or a few hidden factors,
 the PCA prototype filter can offer more power than that of the group knockoff filter. On the other hand, the PCA will not work well if the signals within each group have opposite signs and the signals in each group are canceled almost completely. Due to the improved statistical power and computational efficiency, the prototype knockoff filter with the SDP construction of the knockoff is especially attractive when the equi-correlation construction of the group knockoff gives small $\lambda_{min}$ for certain design matrices. In this case, one may need to use the SDP construction for the group knockoff, which could be expensive.

We have also performed some analysis of the knockoff filter.  Our analysis reveals that certain path method statistics for the knockoff filter may suffer from loss of power for certain design matrices and a specially designed response even if the signal strengths are relatively strong. We provided some partial understanding of this special phenomena and identified the alternating sign effect that could generate this phenomena.
Our numerical results have confirmed that several statistics could lose significant power for certain design matrices and a specially constructed response due to the alternating sign effect.

 \paragraph{Acknowledgments. \\ }

 The research was in part supported by NSF Grants DMS 1318377 and DMS 1613861. The first author's research was conducted during his visit to ACM at Caltech. We are very thankful for Prof. Emmanuel Cand\'es' valuable comments and suggestions to our work. We also thank Prof. Rina Foygel Barber for communicating with us regarding her group knockoff filter and Dr. Lucas Janson for his insightful comments on our PCA prototype filter. We are grateful to the anonymous referees for their valuable comments and suggestions and for pointing out a potential problem in a numerical example in our earlier manuscript using the \textit{glmnet} package in solving the Lasso problem.

 \end{document}